\documentclass[letterpaper, 10 pt, conference]{ieeeconf}  

\IEEEoverridecommandlockouts                              

\usepackage[skip=3pt]{caption}
\usepackage{cite}
\usepackage{amsmath,amssymb,amsfonts}
\usepackage{algorithm}
\usepackage[noend]{algpseudocode}
\usepackage{graphicx}
\usepackage{booktabs}
\usepackage{bbm}
\usepackage{adjustbox,lipsum}
\usepackage{multirow}
\usepackage{supertabular}
\usepackage{textcomp}
\usepackage{setspace}
\usepackage{scalerel}
\usepackage{amssymb}
\usepackage{subcaption}
\usepackage{float}
\usepackage{graphicx}
\usepackage{xcolor}
\usepackage{amsmath}
\usepackage{relsize}
\usepackage{blindtext}
\newtheorem{theorem}{Theorem}
\newcommand\scalemath[2]{\scalebox{#1}{\mbox{\ensuremath{\displaystyle #2}}}}
\newtheorem{lemma}{Lemma}
\newtheorem{Assumption}{Assumption}
\newtheorem{Definition}{Definition}

\DeclareMathOperator{\EX}{\mathbb{E}}
\def\Algo{SORL }

\title{\LARGE \bf
Safety Optimized Reinforcement Learning via Multi-Objective Policy Optimization
}

\author{Homayoun Honari, Mehran Ghafarian Tamizi, Homayoun Najjaran$^{1}$
\thanks{*This work was supported by Apera AI and Mathematics of Information Technology and Complex Systems (MITACS) under IT16412 Mitacs Accelerate.}
\thanks{$^{1}$ The authors are all affiliated with the Faculty of Engineering and Computer Science,
        University of Victoria, Victoria, BC, Canada
        {\tt\small hmnhonari@uvic.ca, mehranght@uvic.ca, najjaran@uvic.ca}}%
}

\begin{document}

\maketitle
\thispagestyle{empty}
\pagestyle{empty}

\begin{abstract}

Safe reinforcement learning (Safe RL) refers to a class of techniques that aim to prevent RL algorithms from violating constraints in the process of decision-making and exploration during trial and error. In this paper, a novel model-free Safe RL algorithm, formulated based on the multi-objective policy optimization framework is introduced where the policy is optimized towards optimality and safety, simultaneously. 
The optimality is achieved by the environment reward function that is subsequently shaped using a safety critic.
The advantage of the \textit{Safety Optimized RL (SORL)} algorithm compared to the traditional Safe RL algorithms is that it omits the need to constrain the policy search space. This allows SORL to find a natural tradeoff between safety and optimality without compromising the performance in terms of either safety or optimality due to strict search space constraints.
Through our theoretical analysis of SORL, we propose a condition for SORL's converged policy to guarantee safety and then use it to introduce an aggressiveness parameter that allows for fine-tuning the mentioned tradeoff.
The experimental results obtained in seven different robotic environments indicate a considerable reduction in the number of safety violations along with higher, or competitive, policy returns, in comparison to six different state-of-the-art Safe RL methods. The results demonstrate the significant superiority of the proposed SORL algorithm in safety-critical applications.
\end{abstract}

\section{Introduction}
Reinforcement learning (RL) is a class of machine learning methods where an agent learns to make decisions by interacting with an environment to maximize rewards. However, the trial-and-error nature of training RL algorithms makes them challenging to use in safety-critical applications where the execution of some actions might lead to system failure. To tackle this, Safe RL algorithms aim to incorporate safety into the learning process to ensure that the policy learned by the algorithm avoids dangerous states. These algorithms have been applied successfully in various real-world domains such as robotics\cite{garcia2020teaching,li2021reinforcement} and autonomous driving\cite{isele2018safe} showing their great potential to enable the control of real-world systems with a minimized total number of failures.

As reviewed extensively in \cite{garcia2015comprehensive,brunke2022safe,gu2022review}, most of the Safe RL methods utilize the Constrained Markov Decision Process (CMDP) framework. Algorithms under this framework specify a level of safety that the policy must adhere to while exploring unknown states and improving its reward performance. However, a major disadvantage of this class of algorithms is their susceptibility to converge to a suboptimal policy 
due to suboptimal tuning of the safety-related hyperparameters.

To mitigate this, our work presents a novel model-free Safe RL algorithm, named Safety Optimized Reinforcement Learning (SORL), designed to enhance both safety and reward performance of the agent, simultaneously. Unlike conventional methods, we tackle the Safe RL problem as a multi-objective policy optimization problem, which allows us to introduce a reward-shaping technique that encourages the agent to explore the environment safely while striving to achieve better performance. This formulation provides an advantage over other model-free Safe RL approaches by eliminating the need to fine-tune the degree of constraining the policy search space (i.e., $\epsilon_{safe}$). As a result, the algorithm will be able to reach a natural trade-off between performance and safety.
Through our analysis of SORL, we guarantee the safety of its converged policy through a condition.
This condition motivates the introduction of the concept of aggressiveness in our algorithm which provides an intuitive way to tune the hyperparameters of the proposed algorithm.

\begin{figure}[t]
\centering
\captionsetup{justification=centering}
\includegraphics[width=1\columnwidth,height=0.55\linewidth]{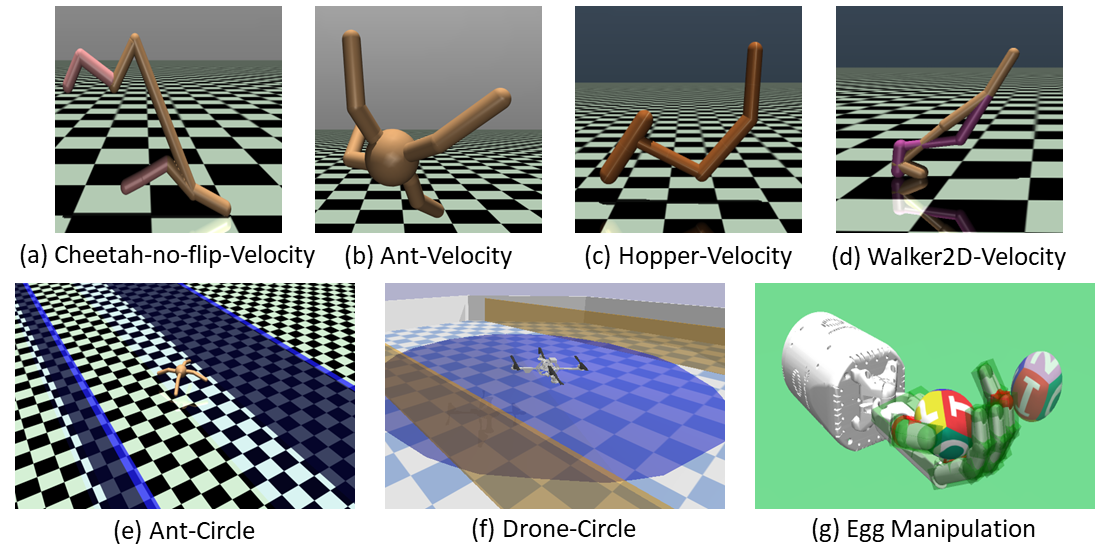}
\caption{Visualization of the simulated safety-concerned robotics environments used to evaluate SORL in the context of three main safety topics of system-level safety (a-d), collision avoidance (e,f) and safe manipulation (g). The top row also showcases some possible constraint violation for the system-level safety benchmarks. 
}
\label{contour}
\vspace{-1.5em}
\end{figure}
Finally, to demonstrate the effectiveness of our proposed algorithm we conduct experiments in seven different safety-concerned simulated robotics problems, which are divided into three main safety topics, namely system-level safety, collision avoidance, and safe manipulation. The tasks are visualized in Fig.~\ref{contour}. The algorithm is compared with six other state-of-the-art model-free Safe RL methods. Our results indicate superior performance in both optimality and safety aspects.
\section{Related work}

Numerous approaches in the literature address safe reinforcement learning. Altman et al.~\cite{altman1999constrained} studied RL algorithms under CMDP framework which aims to train them while satisfying certain constraints.
In this regard, Lagrangian methods~\cite{bertesekas1999nonlinear} are widely employed for efficient CMDP resolution. Shen et al.~\cite{shen2014risk} introduced the risk-sensitive policy optimization (RSPO) algorithm which decreases the LR optimization term to 0 sequentially. Furthermore, Tessler et al.~\cite{tessler2018reward} introduced Reward-constrained policy optimization (RCPO), which optimizes policy and Lagrange multiplier via dual gradient descent.
Furthermore, Zhang et al.~\cite{zhang2020first} proposed the First-Order Policy Optimization (FOCOPS) method, which identifies the optimal update policy through constrained optimization in a nonparametric policy space and subsequently maps it back into the parametric policy space.
Stooke et al.~\cite{stooke2020responsive} incorporated proportional and derivative control into the Lagrange multiplier updates.

Moreover, another approach in the literature to tackle the problem of Safe RL is through modifying unsafe actions locally, meaning they modify the action when it is identified as leading to an unsafe state. Srinivasan et al.~\cite{srinivasan2020learning} uses LR formulation and rejects unsafe policy actions above a safety threshold. Dalal et al.~\cite{dalal2018safe} applies a modification to the reward policy's action using a safety layer that solves its formulation analytically.
Furthermore, Yu et al.~\cite{yu2022towards} proposed SEditor that trains a safety editing policy that modified the selected actions considered as unsafe into safe actions.
Koller et al. \cite{koller2018learning} introduce a learning-based model predictive control (MPC) framework with high probabilities of safety constraint satisfactions which is attained through a Gaussian process statistical model.
Hsu et al.~\cite{hsu2022improving} introduced an unsupervised action planning method that stores the agent's recovery actions for leaving unsafe areas in a dedicated replay buffer, subsequently utilizing it when the agent faces an unsafe state. Safe model-based policy optimization (SMBPO) \cite{thomas2021safe} proactively plans a brief horizon into the future to anticipate and prevents safety violations by applying penalties to unsafe trajectories. Recovery RL \cite{thananjeyan2021recovery} balances exploration and safety through the utilization of either a backup policy, employed for ensuring safety, or MPC to determine the optimal action sequence.

Finally, it is noteworthy that most of the aforementioned methods require the specification of the extent to which the policy is limited, whether using CMDP framework or attempting to detect unsafe actions, which necessitates the careful fine-tuning of the safety threshold.

\section{Preliminaries}
In this section, the basic background concepts and formulations are explained. First, the MDP framework is explained and the safety critic and reward penalty framework are discussed.
\subsection{Markov decision process}
Reinforcement learning is trained under Markov Decision Process (MDP) framework which is presented as $<S, A, P, r, \gamma, \mu>$ and is outlined in \cite{sutton2018reinforcement}. The MDP is composed of a state space $S$, an action space $A$, and a reward function $r: S \times A \times S \mapsto \mathbb{R}$. The transition function $P: S \times A \times S \mapsto [0,1]$ determines the probability $P(s'|s,a)$ of transitioning from state $s$ to $s'$ by executing action $a$. The initial state distribution, $\mu: S \mapsto [0,1]$, and the discount factor, $\gamma \in [0,1)$, are also included. Finally, the policy $\pi: S \mapsto \Delta_{A}$ is the probability distribution over actions, with $\pi(a|s)$ indicating the likelihood of taking action $a$ at state $s$. The value function of a policy $\pi$ for a state-action pair $(s,a)$ and the resulting recursive equation, called the Bellman equation, can be written as:
\begin{equation}
\begin{aligned}
    &Q^{\pi}_{r}(s,a) = \mathbb{E}_{s_t\sim P,a_t\sim\pi}[\sum_{t=0}^\infty \gamma^{t} r(s_t,a_t) |s_0 = s, a_0=a]\\
    &\qquad\qquad=\mathbb{E}_{s'\sim P}[r(s,a)+\gamma V^{\pi}_{r}(s')]
\end{aligned}
\end{equation}
The ultimate objective of an RL algorithm is to maximize the expected discounted cumulative return given the initial state distribution $\mu$:
\begin{equation}
    \pi^*=\underset{\pi \in \Pi}{\text{argmax}} ~J^\pi_r = \underset{\pi \in \Pi}{\text{argmax}} ~\mathbb{E}^{\pi}_{s_0 \sim \mu}[\sum_{t=0}^{\infty}\gamma ^t r^t]
\end{equation}
\subsection{Safety critic}

The safety critic $Q^{\pi}_{safe}$ as described in \cite{srinivasan2020learning}, is based on the safety-aware MDP framework, which is represented as $<S,A,P,r,c,\gamma,\gamma_{safe}>$. 
The safety critic's discount factor is denoted by $\gamma_{safe}$, and the safety signal $c(s)$ is used to determine whether a given state $s$ is safe or not:
\begin{equation}
c(s) = 
\begin{cases}
       1 &\quad \mathrm{if}~ s \in S_{unsafe} \\ 
       0 &\quad \mathrm{otherwise}\\
     \end{cases}
\end{equation}
The main purpose of the safety critic is to estimate the likelihood of a policy failure in the future, based on the expected cumulative discounted probability of failure:

\begin{equation}\label{safebell}
\begin{aligned}
    &Q_{safe}^{\pi}(s,a) = \mathbb{E}_{s_t \sim P, a_t \sim \pi}\big[c(s) + (1 - c(s))\sum_{t = 1}^{\infty} [\gamma_{safe} ^ t c(s_t)]\big] \\
    & = \text{Pr}[c(s) = 1] + \gamma_{safe}\mathbb{E}_{s' \sim P}\big[(1 - c(s))V_{safe}^\pi(s')\big] 
\end{aligned}
\end{equation}
\subsection{Reward penalty framework}\label{rewardpenaltyframework}
Generally, the main purpose of Safe RL algorithms is to identify, and avoid, the set of states that violate the safety constraints, which are refered to as $S_{unsafe}$. To do so, similar to the previous works on Safe RL~\cite{hans2008safe,thomas2021safe}, we can characterize a subset of the state space that are not considered unsafe but will lead to an unsafe state inevitably:
\begin{Definition}
    A state $s\in S$ is considered \textbf{Irrecoverable} if for any sequence of actions $a_0, a_1, a_2, ...$, starting from state $s_0 = s$ and following the transition probability $s_{t+1} \sim P(s_t, a_t)~ \forall t \in \mathbb{N}$, there exists some time step $\bar{t} \in \mathbb{N}$ s.t. $s_{\bar{t}}\in S_{unsafe}$.
\end{Definition}

Naturally, based on the definition, the safe state space $S_{safe}$ encompasses the subset of the state space which are neither categorized as unsafe nor irrecoverable. Correspondingly, an action is considered safe if executing it leads to a safe state. 
Furthermore, we can assume a soon occurrence of safety violation after entering into it an irrecoverable state:

\begin{Assumption}\label{Horizon}
 For any state $s_0\in S_{\mathrm{irrecoverable}}$ and for any sequence of actions $a_0,a_1,\ldots$ starting from $s_0$ and $s_{t+1}\sim P(s_t,a_t) ~\forall t\in\mathbb{N}$, there exists $\bar{t}\in\{1,...,H^*\} \mathrm{~s.t.~} s_{\bar{t}}\in S_{\mathrm{unsafe}}$.
\end{Assumption}

Finally, based on the reward penalty MDP $<S,A,P,\tilde{r},\gamma>$ which is introduced in~\cite{thomas2021safe}, the reward function can be modified as:
\begin{equation}\label{reward}
    \tilde{r}(s,a) = 
     \begin{cases}
       r(s,a) &\quad \text{if } s \notin S_{unsafe}\\
       -C &\quad \text{otherwise} \\ 
     \end{cases}
\end{equation}

where C satisfies the following inequality:
\begin{equation}\label{CBound}
    C > \frac{r_{max} - r_{min}}{\gamma ^ {H^*}} - r_{max}
\end{equation}
The terminal state cost $C \in \mathbb{R}$ is used to penalize the RL agent when unsafe trajectories are executed.

\section{Method}
In this work, a model-free safe reinforcement learning algorithm is proposed which uses safety signals to avoid unsafe regions. In the training process, the safety critic is used to modify the reward function such that exploration in unsafe regions is prevented.

\subsection{Multi-Objective Policy Optimization}
Unlike the conventional Safe RL setting where the algorithm is trained under the Constrained MDP framework, we formulate \Algo in a multi-objective policy optimization setting \cite{abdolmaleki2020distributional,abdolmaleki2021multi}. To this end, we propose \textit{Safety-aware reward penalty MDP} $<S,A,P,\tilde{r},c,\gamma,\gamma_{safe}>$. Under this setting, in multi-objective policy optimization, there are multiple (often conflicting) objectives and the aim is to optimize the objectives simultaneously. For this purpose, the policy performance is defined as a 2-dimensional vector:

\begin{equation}\label{perf}
J(\pi)= 
\begin{bmatrix} J_r(\pi) \\  
J_c(\pi) \end{bmatrix}=
\begin{bmatrix} \EX_{\tau\sim\pi}[\sum_{t=0}^{\infty} \gamma^t r_t] \\  
\EX_{\tau\sim\pi}[\sum_{t=0}^{\infty} \gamma^t c_t] \end{bmatrix}
\end{equation}

Furthermore, according to the Multi-Objective optimization literature, the dominance of a policy $\pi$ relative to $\pi'$ can be defined as:

\begin{equation}\label{dominn}
\mathrm{if} \;\;\;  \begin{cases} J_r(\pi)\geq J_r(\pi') \\\mathbf{\qquad and}\\ J_c(\pi)\leq J_c(\pi') \end{cases}
\Rightarrow
J(\pi)\succcurlyeq J(\pi')
\end{equation}

The field of multi-objective deep reinforcement learning is an active area of research that encompasses various methodologies aimed at optimizing the expected return of a set of potentially conflicting reward functions \cite{roijers2013survey}. One widely-used approach for addressing this challenge involves scalarizing the reward vector using a scalarization function, which enables the optimization of the scalarized function and, subsequently, the optimization of all rewards. In the context of SORL, the objective is to devise a reward-shaping scheme ($\hat{r}$) that optimizes both performance functions by optimizing a single policy:
\begin{equation}\label{domin}
    J_{\hat{r}}(\pi)\geq J_{\hat{r}}(\pi')\Rightarrow\pi\succcurlyeq\pi'
\end{equation}

\subsection{Safety Optimized Reward Shaping}
To address Eq.~\ref{domin}, the notion of the reward function as a vector consisting of the actual reward function $r$ and the safety signal function $c$ seems plausible. However, it is not possible to use the safety signal directly in the scalarization function as it is a sparse function and its immediate value does not bear much significance. To this end, we propose the augmented reward function, based on the safety estimate of the safety critic, defined as:

\begin{equation}\label{rtilde}
\begin{aligned}
    &\tilde{r} (s_t,a_t)= 
    \begin{cases}
    [1-\lambda Q_{safe}^\pi(s_t,a_t)]r(s_t,a_t),& \text{if } r(s_t,a_t)\geq 0\\
     \lambda Q_{safe}^\pi(s_t,a_t)r(s_t,a_t),              & \text{otherwise}
    \end{cases}
\end{aligned}
\end{equation}

where $\lambda>0$ is defined as the safety critic significance factor. Finally, under the safety-aware reward penalty MDP framework, the final reward shaping will be as follows:

\begin{equation}\label{MCreward}
    \hat{r} (s_t,a_t)=
    \begin{cases}
        \Tilde{r} (s_t,a_t),& \text{if } s_{t+1}\in S_{safe} \\
        -C,& \text{otherwise}
    \end{cases}
\end{equation}

\subsection{Safety Guarantee}\label{guarantee}
The aim of this section is to guarantee the safety of the converged policy when Eq.~\ref{MCreward} reward shaping scheme is used which will help us during the hyperparameter tuning phase.
To prove the theorem, the following assumptions are used:
\begin{Assumption}\label{reward_bound}
    The reward function $r$ is bounded in the range $[r_{min},r_{max}]$ where $r_{min}<0$ and $r_{max}>0$.
\end{Assumption}
\begin{Assumption}\label{deterministic}
    The environment is deterministic in terms of the safety violations. In other words, for any state $s\in S$, $\mathrm{Pr}[c(s) = 1]$ is either 0 or 1.
\end{Assumption}
In order to provide a safety guarantee we must first study the irrecoverable states discussed in section \ref{rewardpenaltyframework}. Using Assumption \ref{Horizon} allows for further categorizing the irrecoverable states into levels of unsafety. 

\begin{lemma}\label{irreclemma}
    In an environment where Assumption \ref{Horizon} holds, for any trajectory $\tau=\{(s_0,a_0),...,(s_{|\tau|},a_{|\tau|})\}$ where $s_0\in S_{irrecoverable}$, $s_t\sim P(s_{t-1},a_{t-1})$, and $s_{|\tau+1|}\sim P(s_{|\tau|},a_{|\tau|})\in S_{unsafe}$, and for any $t\in\{1,...,{|\tau|}\}$ we have:
    \begin{equation}
        Q_{safe}^\pi(s_t,a_t)\geq {(\gamma_{safe})}^{H^*-t}
    \end{equation}
\end{lemma}
\proof
Recall that from Assumption \ref{Horizon} we know the length of the trajectory is ${|\tau|}\leq H^*$. Moreover, at any time step $t$ and its corresponding state $s_t$, consider the space of all the trajectories that start from $s_t$: $\mathrm{T}_t=\{\tau^{\prime}:(s^{\prime}_0=s_t,a^{\prime}_0),\ldots,(s^{\prime}_{|\tau^{\prime}|},a^{\prime}_{|\tau^{\prime}|}),s^{\prime}_{|\tau^{\prime}+1|}\sim P(s_{|\tau^{\prime}|},a_{|\tau^{\prime}|})\in S_{unsafe}\}$.
Suppose there exists a trajectory $\tau^{\prime}\in\mathrm{T}_t$ with $|\tau^{\prime}|>H^*-t$. Consequently, the concatenation $\tau(s_0:s_t)\cup\tau^{\prime}$ will have a length greater than $H^*$ which will violate Assumption \ref{Horizon} since there will then exist a trajectory from the irrecoverable state $s_0$ with a length greater than~$H^*$. \\
Therefore, based on the original trajectory $\tau$, we can conclude that $\underset{\tau^{\prime}\in\mathrm{T}_t}{\max}|\tau^{\prime}|\leq H^*-t$. Finally, by the definition of the safety critic in Eq.~\ref{safebell} and using Assumption \ref{deterministic}, the lower bound can be achieved.
\endproof
It is noteworthy that to establish safety bounds for the safety critic's evaluation in safe states, we need to make assumptions about the policy's behavior and its interactions with the environment because in some cases, the policy may choose unsafe actions even in safe states, making the safety critic's output an expectation that includes both safe and unsafe actions. Therefore, to maintain a broad analysis, we avoid making such limiting assumptions.
\begin{theorem}\label{the}
Under the safety-aware reward penalty MDP framework, let the following condition for safety critic significance factor $\lambda$ hold:
\begin{equation}\label{theoremlambda}
\begin{aligned}
    &\frac{r_{max}}{1-\gamma}-\frac{r_{max}+C}{1-\gamma}\gamma^{|\tau_{uwc}^{*}|}\\
    &<(\frac{\gamma_{safe}^{H^*}r_{max}}{1-\frac{\gamma}{\gamma_{safe}}}-\frac{\gamma_{safe}^{H^*}r_{max}}{1-\frac{\gamma}{\gamma_{safe}}}(\frac{\gamma}{\gamma_{safe}})^{|\tau_{uwc}^{*}|}+\frac{\gamma_{safe}r_{min}}{1-\gamma})\lambda
\end{aligned}
\end{equation}
where C and $|\tau_{uwc}^{*}|$ follow Eq.~\ref{CBound} and Eq.~\ref{optimaltraj}, respectively.
Consequently, for any state $s$ we have: $\hat{Q}^*(s,a)>\hat{Q}^*(s,a')$, where action $a$ is safe, $a'$ is unsafe, and $\hat{Q}^*$ is the $Q^*$ value-function following Eq.~\ref{MCreward} reward-shaping. \end{theorem}
\begin{proof} By Assumption~\ref{Horizon}, if $a'$ is unsafe, it is going to lead to an unsafe state in at most $H^*$ steps. Therefore, by Assumption \ref{reward_bound}, the maximum discounted return in the worst-case scenario can be expressed as a function of the length of the trajectory before reaching the unsafe state:
\begin{equation}\label{max_return}
    \begin{aligned}
    &R^{\pi}_{wc}(|\tau|)=\\
    &\qquad\sum_{t=0}^{|\tau|-1}(\gamma^t\big[1-\lambda Q_{safe}^\pi(s_t,a_t)\big]r_{max})+\sum_{t=|\tau|}^\infty(\gamma^t(-C))
    \end{aligned}
\end{equation}
While it is possible to ignore the safety critic term and upper bound the equation, we provide a tighter upper bound in our case. We first use Lemma \ref{irreclemma} to upper bound the function (in the following we use the notation $x$ to indicate the variability of $|\tau|$):
\begin{equation}
    \begin{aligned}
    &R_{wc}^\pi(x)\leq\sum_{t=0}^{x-1}(\gamma^t\big[1-\lambda \gamma_{safe}^{H^*-t}\big]r_{max})+\sum_{t=x}^\infty(\gamma^t(-C))
    \\ &=(\frac{r_{max}}{1-\gamma}-\frac{\lambda\gamma_{safe}^{H^*}r_{max}}{1-\frac{\gamma}{\gamma_{safe}}})-\frac{r_{max}+C}{1-\gamma}\gamma^x
    \\&\qquad\qquad\qquad\qquad+\frac{\lambda\gamma_{safe}^{H^*}r_{max}}{1-\frac{\gamma}{\gamma_{safe}}}(\frac{\gamma}{\gamma_{safe}})^x=R_{uwc}(x)
    \end{aligned}
\end{equation}
To find the maximum unsafe return in the domain $x\in [1, H^*]$, we take the derivative with respect to $x$ and set it to zero: $\partial R_{uwc}^{\pi}(x) / \partial x = 0$. By solving this derivative and ensuring that the trajectory length remains within the range $|\tau_{uwc}^{}| \in [1, H^*]$, we can determine the trajectory length with the highest return as:
\begin{equation}\label{optimaltraj}
    |\tau_{uwc}^{*}|=
    \begin{cases}
    \scalemath{0.7}{\left\lfloor \frac{\mathlarger{\ln}({\dfrac{\lambda\gamma_{safe}^{H^*}r_{max}}{1-\dfrac{\gamma}{\gamma_{safe}}}}\mathlarger{\ln}{\dfrac{\gamma}{\gamma_{safe}}})-\mathlarger{\ln}({\dfrac{r_{max}+C}{1-\gamma}}\mathlarger{\ln}{\gamma})}{\mathlarger{\ln}{\mathlarger{\gamma_{safe}}}}\right\rfloor} & \mathrm{if }\in[1,H^*]\\
    \qquad\quad\underset{|\tau|\in\{1,H^*\}}{\mathrm{argmax}}R_{uwc}(|\tau|) & \mathrm{otherwise}
    \end{cases}
\end{equation}
Hence, Eq.~\ref{max_return} can be upper bounded as:
\begin{equation}
    R_{wc}^\pi(|\tau|)\leq R_{uwc}(|\tau_{uwc}^{*}|)
\end{equation}
Furthermore, executing the safe action $a$ leads to a safe state where a safe trajectory can be generated which does not encounter a safety violation. The discounted return with the minimum reward (Assumption~\ref{reward_bound}) can be lower bounded as:
\begin{equation}\label{rmin}
\begin{aligned}
    &\sum_{t=0}^\infty(\gamma^t\lambda Q_{safe}^\pi(s_t,a_t)r_{min})=
    \lambda r_{min}\sum_{t=0}^\infty(\gamma^t Q_{safe}^\pi(s_t,a_t))\\
    &\qquad\qquad\quad\qquad\geq \lambda r_{min}\gamma_{safe}\sum_{t=0}^\infty(\gamma^t)= \frac{\lambda\gamma_{safe} r_{min}}{1-\gamma}
\end{aligned}
\end{equation}
To establish the inequality, we leverage the observation that $Q_{safe}^\pi(s_t, a_t) \leq \gamma_{safe}$. This is based on the
fact that for any timestep $t$, execution of the state-action pair $(s_{t-1},a_{t-1})$ leads to the safe state $s_{t}$.
Hence, Assumption~\ref{deterministic} ensures us that $\mathrm{Pr}[c(s_t)=1]=0$; consequently, the right hand side equation in Eq.~\ref{safebell} can be simplified as $\gamma_{safe}\mathbb{E}_{s' \sim P}[V_{safe}^\pi(s')]$. Finally, because the expectation operator outputs a value within the range of zero and one, the observation is justified. \\
Therefore, since the discounted return of staying safe forever must always be higher than a trajectory that has a safety violation, it suffices that:
\begin{equation}\label{case1ineq}
    \begin{aligned}
        \Delta=\frac{\lambda\gamma_{safe} r_{min}}{1-\gamma}-R_{uwc}(|\tau_{uwc}^{*}|)>0
    \end{aligned}
\end{equation}
Rearranging Eq.~\ref{case1ineq} gives us the condition.
\end{proof}

It is possible to derive bounds similar to Theorem \ref{the} with other variations of Assumption \ref{reward_bound}. Moreover, intuitively, the value of $\Delta$ in Eq.~\ref{case1ineq} determines the degree of aggressiveness of the algorithm. The closer this value is to zero, the more the algorithm prioritizes optimality and performance over its failure rate. Hence, for each task, we define the level of aggressiveness in SORL by specifying the value of $\Delta$ and fine-tuning $\lambda$ to align it as closely as possible with the specified value.

\begin{algorithm}[t]
\caption{Safety Optimized RL}\label{SORL}
\begin{algorithmic}[1]
\Require Safety critic significance factor $\lambda$ and $\Delta$, horizon $H^*$
\State \textbf{Initialize} policy $\pi_\theta$, critic $Q_{\phi_1}, Q_{\phi_2}$, replay buffer $\mathcal{D}$,
safety critic ${Q}_{safe}^{\psi_1}, {Q}_{safe}^{\psi_2}$, replay buffer $D_{safe}$
\For{$e=1,...,E_{max}$}
    \State $s_1 \gets env.reset()$
    \For{$t=1,...,T_{max}$}
        \State $a_t \sim \pi_{\theta}(.|s_t)$
        \State $\hat{c}_t \gets \max\{{Q}_{safe}^{\psi_1}(a_t|s_t), {Q}_{safe}^{\psi_2}(a_t|s_t)\}$
        \State $s_{t+1}, r_t, c_t, done \gets env.step(a_t)$
        \State compute empirical $r_{min},r_{max}$ and update $C$
        \State solve and update $\lambda$ (Eq.~\ref{theoremlambda})
        \If{$c_t==0$}
            \If{$r_t>0$} $\hat{r}_t\gets \big[1-\lambda\hat{c}_t\big]r_t$
            \Else {} $\hat{r}_t\gets \lambda\hat{c}_tr_t$
            \EndIf
        \Else
            \State $\hat{r}_t\gets-C$
            \State $\mathcal{D}_{safe}\gets \mathcal{D}_{safe}\cup(s_t,a_t,c_t,\hat{r}_t,s_{t+1})$
        \EndIf
        \State $\mathcal{D}\gets \mathcal{D}\cup(s_t,a_t,c_t,\hat{r}_t,s_{t+1})$
        \State train $\pi_\theta, Q_{\phi_1}, Q_{\phi_2}$ on $\mathcal{D}$
        \openup 1mm \State train ${Q}_{safe}^{\psi_1}, {Q}_{safe}^{\psi_2}$ on $\mathcal{D}\cup \mathcal{D}_{safe}$
        \If{done} Break
        \EndIf
    \EndFor
\EndFor
\end{algorithmic}
\end{algorithm}

\subsection{Safety Optimized Reinforcement Learning Algorithm}\label{algosec}
The training process of the SORL algorithm is presented in Algorithm \ref{SORL}. The proposed algorithm can be built on top of any Model-Free RL algorithm and uses two safety critic networks to estimate the cumulative discounted probability of failure. Two replay buffers are utilized, one for storing all the transitions and the other one for storing unsafe transitions. In the training process, the range of the reward function is computed empirically and, as discussed in Section \ref{guarantee}, we update the value of $\lambda$ to satisfy Eq.~\ref{theoremlambda} and the predefined $\Delta$ value in Eq.~\ref{case1ineq}. It should be noted that in the case where $|\tau_{uwc}^{*}|\in(1,H^*)$, $\Delta$ becomes a non-linear equality which makes it challenging to derive a closed-form solution. Therefore, in practice, given an initial value for $\lambda$, we find a solution in the locality of it that satisfies the conditions.

\begin{figure*}[t]
\centering
\captionsetup{justification=centering}
\includegraphics[width=2.1\columnwidth,height=0.4\linewidth]{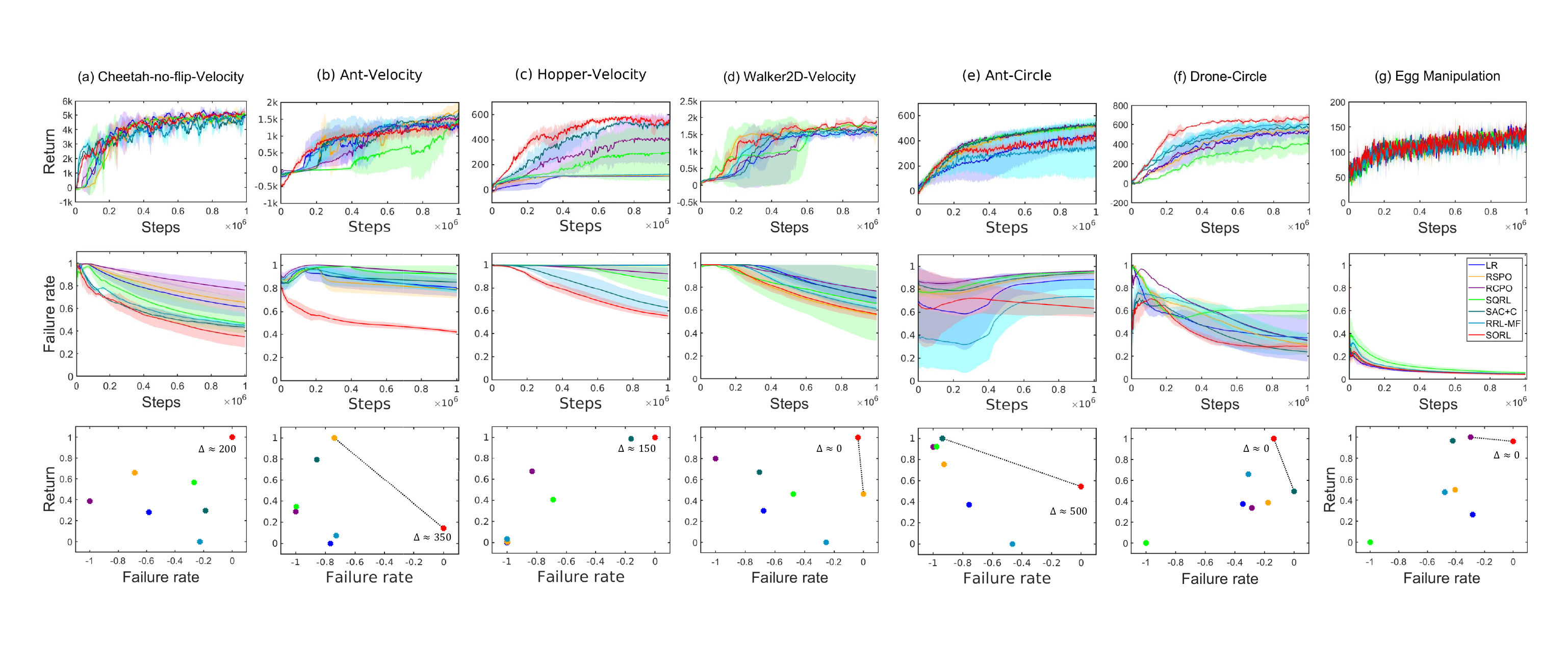}

\caption{
Benchmark results of SORL compared with six other Safe RL algorithms. \textbf{(Top row):} Return values achieved during the training phase (higher is better). \textbf{(Middle row):} Episodic failure rates suffered during the training phase (lower is better). \textbf{(Bottom row):} Pareto optimality plot corresponding to the return and failure rate values (closer to the top right corner is better). They also display the specific value of $\Delta$ employed for executing SORL. For easier comparison, the return and failure rate values are normalized and the failure rate is scaled to lie within -1 and 0. The Pareto optimality solutions are highlighted through the dotted line. 
}
\label{viol}
\vspace{-2.em}
\end{figure*}

\section{Experiment}
In the following section, the performance of SORL is studied. Particularly, we aim to investigate two questions:
\begin{itemize}
    \item How does the safety formulation performs compared with the other model-free off-policy Safe RL algorithms?
    \item How does the value of $\Delta$ in Eq.~\ref{case1ineq} affect the performance of SORL?
\end{itemize}
\subsection{Benchmarks and Comparison Methods}
In order to evaluate the level of safety that the model can achieve, we execute it on three safety-concerned categories of environments:
\begin{itemize}
    \item \textbf{System-level safety:} RL algorithms are often used to optimally control robots while adhering to the system limits.
    We assess our proposed model using four MuJoco environments: \textit{Cheetah-no-flip-Velocity, Ant-Velocity, Hopper-Velocity,} and \textit{Walker2D-Velocity}. In these environments, the agent must learn to move faster in the x-direction while avoiding actions that cause the robot to fall and fail. Additionally, safety violations occur if the robot exceeds a certain velocity. We obtained the codebase for Hopper-Velocity, Walker2D-Velocity, and Ant-Velocity from~\cite{ji2024safety}. For Cheetah-no-flip-Velocity, the base environment was adopted from~\cite{thomas2021safe} and the Velocity constraint was added to it.
    \item \textbf{Collision Avoidance:} Besides the inherent robot limits, additional constraints from the environment can impact the algorithm. Collision avoidance is one such constraint, where the controller aims to control the robot while preventing collisions with obstacles. We evaluate our algorithms using Ant-Circle \cite{achiam2017constrained} and Drone-Circle (from Bullet Safety Gym codebase~\cite{Gronauer2022BulletSafetyGym}) environments. These assessments involve controlling robots to move in circular paths while staying within a safety region smaller than the circle's radius.
    \item \textbf{Safe Manipulation:} Finally, one of the important applications of Safe RL is safe manipulation. To this end, we adopt a modified version of the in-hand object manipulation from Gymnasium Robotics \cite{gymnasium_robotics2023github} which uses a dexterous hand to manipulate an egg to achieve a target pose. In this task, if the hand exerts a normal force more than a threshold (20 N), the egg will get crushed and the agent will fail.
\end{itemize}
The episode ends whenever a safety violation has been incurred. Moreover, in all the environments, the alive bonus has been eliminated to evaluate the performance of the safety algorithms in situations where the original reward shaping does not explicitly encode safety.

Six model-free Safe RL algorithms are used to showcase the performance of SORL.
The comparison algorithms include Lagrangian Relaxation (LR), Safety Q-Functions for RL (SQRL)~\cite{srinivasan2020learning}, Model-Free Recovery RL (RRL-MF)~\cite{thananjeyan2021recovery}, Risk Sensitive Policy Optimization (RSPO)~\cite{shen2014risk}, and Reward Constrained Policy Optimization (RCPO)~\cite{tessler2018reward}. Finally, to study the safety performance of SORL reward shaping, SAC+C is executed which uses Eq.~\ref{reward} reward scheme.
\subsection{Implementation settings}
The codebase for the comparison methods are adopted from~\cite{thananjeyan2021recovery} codebase, and, to have a fair comparison between the algorithms and evaluate their safety, pretraining is disabled for all of the algorithms. To this end, SORL and the comparison methods are built on top of the Soft Actor-Critic algorithm~\cite{haarnoja2018soft}, and, for fair comparison, the general and common hyperparameters of all the algorithms are kept the same. In addition to that, with the help from the problem-specific hyperparamter settings discussed in~\cite{thomas2021safe}, we tune the parameters of the comparison algorithms for each benchmark problem. Moreover, for the hyperparameters relating to SORL, we set $H^*=10$ for all the cases and tune the proposed algorithm based on the value of $\Delta$. For each task, the value of $\Delta$ used to execute SORL is shown in Fig.~\ref{viol}. During our experimentation of various $\Delta$ values within the specified environments, we observed significant variations in the performance of SORL when its change of value is near the magnitude of 50. The magnitude of change in $\Delta$ can be largely attributed to the choice of $\gamma$ and $\gamma_{safe}$.
The results illustrate the mean and variance of the execution of the algorithms with independent random seeds.
\subsection{Results}
For the performance comparison of the different Safe RL algorithms, as depicted in Fig.~\ref{viol}, we report the reward performance and the failure rate of the algorithms. Based on the learning curves, we plot the pareto optimality of the algorithms based on Eq.~\ref{domin}. Our results show dominant and superior performance of SORL in both aspects in Fig.~\ref{viol}(a) and~\ref{viol}(c).
Furthermore, in Fig.~\ref{viol}(b),~\ref{viol}(e), and~\ref{viol}(g)
we can see that the proposed algorithm attains significantly better safety performance while achieving comparable returns. 
Importantly, the suboptimality of the converged policies of the comparison methods in Fig.~\ref{viol}(b) and~\ref{viol}(c)
in one or both aspects of performance can be seen. 
Finally, the results in Fig.~\ref{viol}(d) and \ref{viol}(f) illustrate SORL's consistently higher returns while also maintaining near-dominant safety performance. Thus, we observe that SORL can strike a great balance between safety and optimality offering a great novel solution for safe performance among Safe RL algorithms.
\vspace{-1.em}
\begin{figure}[htb]
\captionsetup{justification=centering}
\centerline{\includegraphics[width=.9\columnwidth,height=0.4\columnwidth]{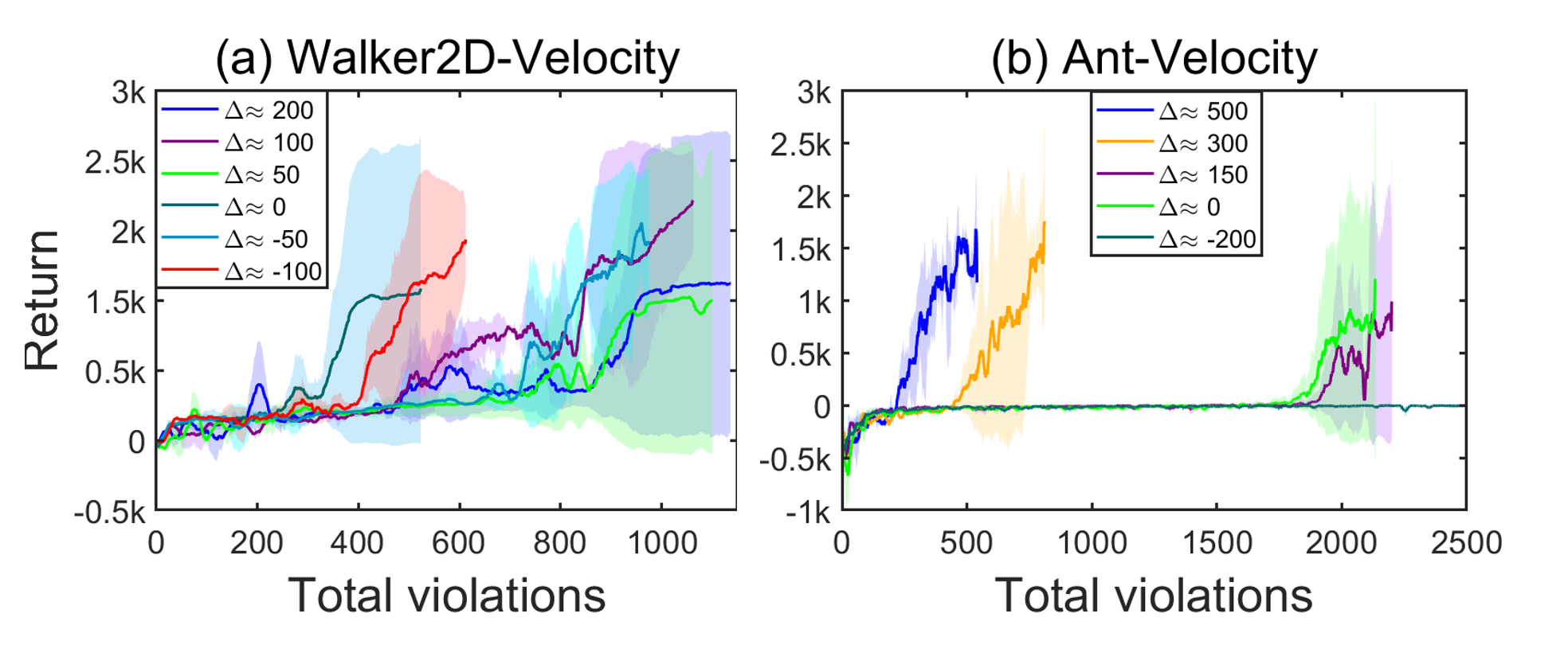}}
\caption{Undiscounted return of the policy versus the total number of violations during the training phase.
}
\label{eps}
\vspace{-1.5em}
\end{figure}
\subsection{Ablation Analysis}\label{ablation}
In this section, we study the effect of $\Delta$ on SORL's performance in two environments. We report the undiscounted return of the policy the algorithm achieves whenever a safety violation has occurred which can show sample efficiency of the algorithm in terms of safety. To gain a better understanding of the effect of $\Delta$, we also chose negative values to study its performance under too aggressive specifications. As depicted in Fig.~\ref{viol}, while being aggressive in Walker2D-Velocity helps SORL attain higher returns in lower number of constraint violations, better performance in Ant-Velocity requires more conservativeness. This may be due to the difference in the dynamics of the robots and their constraints since Ant-Circle also requires a more conservative $\Delta$ value. The dynamics differences is more evident in the comparison between $(-100,-50)$ (where there is no safety guarantee) with $\Delta\approx(500,300)$ in their respective tasks which indicates while being a little more aggressive in one helps in the reduction of the number of constraint violations, the opposite holds true for the other task.
\section{Conclusions}
This paper focuses on the problem of safe exploration and decision-making for RL agents. A novel Safe RL approach based on multi-objective policy optimization framework was proposed which optimized the policy toward optimality and safety, simultaneously. Through theoretical analysis, the safety of SORL's converged policy was guaranteed through a condition which allowed the introduction of the concept of aggressiveness. The concept provided an intuitive way to tune SORL's safety-related hyperparameter. Finally, three main safety topics (viz., system-level safety, collision avoidance, and safe manipulation) were studied through seven different tasks in total.
We evaluated reward and safety performance of the proposed algorithm against six other state-of-the-art model-free Safe RL approaches. The results showed SORL's great capability in attaining better safety performance while achieving better or comparable returns.

\bibliographystyle{IEEEtran}

\begin{thebibliography}{10}
\providecommand{\url}[1]{#1}
\csname url@samestyle\endcsname
\providecommand{\newblock}{\relax}
\providecommand{\bibinfo}[2]{#2}
\providecommand{\BIBentrySTDinterwordspacing}{\spaceskip=0pt\relax}
\providecommand{\BIBentryALTinterwordstretchfactor}{4}
\providecommand{\BIBentryALTinterwordspacing}{\spaceskip=\fontdimen2\font plus
\BIBentryALTinterwordstretchfactor\fontdimen3\font minus
  \fontdimen4\font\relax}
\providecommand{\BIBforeignlanguage}[2]{{%
\expandafter\ifx\csname l@#1\endcsname\relax
\typeout{** WARNING: IEEEtran.bst: No hyphenation pattern has been}%
\typeout{** loaded for the language `#1'. Using the pattern for}%
\typeout{** the default language instead.}%
\else
\language=\csname l@#1\endcsname
\fi
#2}}
\providecommand{\BIBdecl}{\relax}
\BIBdecl

\bibitem{garcia2020teaching}
J.~Garc{\'\i}a and D.~Shafie, ``Teaching a humanoid robot to walk faster
  through safe reinforcement learning,'' \emph{Engineering Applications of
  Artificial Intelligence}, vol.~88, p. 103360, 2020.

\bibitem{li2021reinforcement}
Z.~Li, X.~Cheng, X.~B. Peng, P.~Abbeel, S.~Levine, G.~Berseth, and K.~Sreenath,
  ``Reinforcement learning for robust parameterized locomotion control of
  bipedal robots,'' in \emph{2021 IEEE International Conference on Robotics and
  Automation (ICRA)}.\hskip 1em plus 0.5em minus 0.4em\relax IEEE, 2021, pp.
  2811--2817.

\bibitem{isele2018safe}
D.~Isele, A.~Nakhaei, and K.~Fujimura, ``Safe reinforcement learning on
  autonomous vehicles,'' in \emph{2018 IEEE/RSJ International Conference on
  Intelligent Robots and Systems (IROS)}.\hskip 1em plus 0.5em minus
  0.4em\relax IEEE, 2018, pp. 1--6.

\bibitem{garcia2015comprehensive}
J.~Garc{\i}a and F.~Fern{\'a}ndez, ``A comprehensive survey on safe
  reinforcement learning,'' \emph{Journal of Machine Learning Research},
  vol.~16, no.~1, pp. 1437--1480, 2015.

\bibitem{brunke2022safe}
L.~Brunke, M.~Greeff, A.~W. Hall, Z.~Yuan, S.~Zhou, J.~Panerati, and A.~P.
  Schoellig, ``Safe learning in robotics: From learning-based control to safe
  reinforcement learning,'' \emph{Annual Review of Control, Robotics, and
  Autonomous Systems}, vol.~5, pp. 411--444, 2022.

\bibitem{gu2022review}
S.~Gu, L.~Yang, Y.~Du, G.~Chen, F.~Walter, J.~Wang, Y.~Yang, and A.~Knoll, ``A
  review of safe reinforcement learning: Methods, theory and applications,''
  \emph{arXiv preprint arXiv:2205.10330}, 2022.

\bibitem{altman1999constrained}
E.~Altman, \emph{Constrained Markov decision processes: stochastic
  modeling}.\hskip 1em plus 0.5em minus 0.4em\relax Routledge, 1999.

\bibitem{bertesekas1999nonlinear}
D.~Bertesekas, ``Nonlinear programming. athena scientific,'' \emph{Belmont,
  Massachusetts}, 1999.

\bibitem{shen2014risk}
Y.~Shen, M.~J. Tobia, T.~Sommer, and K.~Obermayer, ``Risk-sensitive
  reinforcement learning,'' \emph{Neural computation}, vol.~26, no.~7, pp.
  1298--1328, 2014.

\bibitem{tessler2018reward}
C.~Tessler, D.~J. Mankowitz, and S.~Mannor, ``Reward constrained policy
  optimization,'' \emph{arXiv preprint arXiv:1805.11074}, 2018.

\bibitem{zhang2020first}
Y.~Zhang, Q.~Vuong, and K.~Ross, ``First order constrained optimization in
  policy space,'' \emph{Advances in Neural Information Processing Systems},
  vol.~33, pp. 15\,338--15\,349, 2020.

\bibitem{stooke2020responsive}
A.~Stooke, J.~Achiam, and P.~Abbeel, ``Responsive safety in reinforcement
  learning by pid lagrangian methods,'' in \emph{International Conference on
  Machine Learning}.\hskip 1em plus 0.5em minus 0.4em\relax PMLR, 2020, pp.
  9133--9143.

\bibitem{srinivasan2020learning}
K.~Srinivasan, B.~Eysenbach, S.~Ha, J.~Tan, and C.~Finn, ``Learning to be safe:
  Deep rl with a safety critic,'' \emph{arXiv preprint arXiv:2010.14603}, 2020.

\bibitem{dalal2018safe}
G.~Dalal, K.~Dvijotham, M.~Vecerik, T.~Hester, C.~Paduraru, and Y.~Tassa,
  ``Safe exploration in continuous action spaces,'' \emph{arXiv preprint
  arXiv:1801.08757}, 2018.

\bibitem{yu2022towards}
H.~Yu, W.~Xu, and H.~Zhang, ``Towards safe reinforcement learning with a safety
  editor policy,'' \emph{Advances in Neural Information Processing Systems},
  vol.~35, pp. 2608--2621, 2022.

\bibitem{koller2018learning}
T.~Koller, F.~Berkenkamp, M.~Turchetta, and A.~Krause, ``Learning-based model
  predictive control for safe exploration,'' in \emph{2018 IEEE conference on
  decision and control (CDC)}.\hskip 1em plus 0.5em minus 0.4em\relax IEEE,
  2018, pp. 6059--6066.

\bibitem{hsu2022improving}
H.-L. Hsu, Q.~Huang, and S.~Ha, ``Improving safety in deep reinforcement
  learning using unsupervised action planning,'' in \emph{2022 International
  Conference on Robotics and Automation (ICRA)}.\hskip 1em plus 0.5em minus
  0.4em\relax IEEE, 2022, pp. 5567--5573.

\bibitem{thomas2021safe}
G.~Thomas, Y.~Luo, and T.~Ma, ``Safe reinforcement learning by imagining the
  near future,'' \emph{Advances in Neural Information Processing Systems},
  vol.~34, pp. 13\,859--13\,869, 2021.

\bibitem{thananjeyan2021recovery}
B.~Thananjeyan, A.~Balakrishna, S.~Nair, M.~Luo, K.~Srinivasan, M.~Hwang, J.~E.
  Gonzalez, J.~Ibarz, C.~Finn, and K.~Goldberg, ``Recovery rl: Safe
  reinforcement learning with learned recovery zones,'' \emph{IEEE Robotics and
  Automation Letters}, vol.~6, no.~3, pp. 4915--4922, 2021.

\bibitem{sutton2018reinforcement}
R.~S. Sutton and A.~G. Barto, \emph{Reinforcement learning: An
  introduction}.\hskip 1em plus 0.5em minus 0.4em\relax MIT press, 2018.

\bibitem{hans2008safe}
A.~Hans, D.~Schneega{\ss}, A.~M. Sch{\"a}fer, and S.~Udluft, ``Safe exploration
  for reinforcement learning.'' in \emph{ESANN}.\hskip 1em plus 0.5em minus
  0.4em\relax Citeseer, 2008, pp. 143--148.

\bibitem{abdolmaleki2020distributional}
A.~Abdolmaleki, S.~Huang, L.~Hasenclever, M.~Neunert, F.~Song, M.~Zambelli,
  M.~Martins, N.~Heess, R.~Hadsell, and M.~Riedmiller, ``A distributional view
  on multi-objective policy optimization,'' in \emph{International Conference
  on Machine Learning}.\hskip 1em plus 0.5em minus 0.4em\relax PMLR, 2020, pp.
  11--22.

\bibitem{abdolmaleki2021multi}
A.~Abdolmaleki, S.~H. Huang, G.~Vezzani, B.~Shahriari, J.~T. Springenberg,
  S.~Mishra, D.~TB, A.~Byravan, K.~Bousmalis, A.~Gyorgy \emph{et~al.}, ``On
  multi-objective policy optimization as a tool for reinforcement learning,''
  \emph{arXiv preprint arXiv:2106.08199}, 2021.

\bibitem{roijers2013survey}
D.~M. Roijers, P.~Vamplew, S.~Whiteson, and R.~Dazeley, ``A survey of
  multi-objective sequential decision-making,'' \emph{Journal of Artificial
  Intelligence Research}, vol.~48, pp. 67--113, 2013.

\bibitem{ji2024safety}
J.~Ji, B.~Zhang, J.~Zhou, X.~Pan, W.~Huang, R.~Sun, Y.~Geng, Y.~Zhong, J.~Dai,
  and Y.~Yang, ``Safety gymnasium: A unified safe reinforcement learning
  benchmark,'' \emph{Advances in Neural Information Processing Systems},
  vol.~36, 2024.

\bibitem{achiam2017constrained}
J.~Achiam, D.~Held, A.~Tamar, and P.~Abbeel, ``Constrained policy
  optimization,'' in \emph{International conference on machine learning}.\hskip
  1em plus 0.5em minus 0.4em\relax PMLR, 2017, pp. 22--31.

\bibitem{Gronauer2022BulletSafetyGym}
S.~Gronauer, ``Bullet-safety-gym: A framework for constrained reinforcement
  learning,'' mediaTUM, Tech. Rep., 2022.

\bibitem{gymnasium_robotics2023github}
\BIBentryALTinterwordspacing
R.~de~Lazcano, K.~Andreas, J.~J. Tai, S.~R. Lee, and J.~Terry, ``Gymnasium
  robotics,'' 2023. [Online]. Available:
  \url{http://github.com/Farama-Foundation/Gymnasium-Robotics}
\BIBentrySTDinterwordspacing

\bibitem{haarnoja2018soft}
T.~Haarnoja, A.~Zhou, K.~Hartikainen, G.~Tucker, S.~Ha, J.~Tan, V.~Kumar,
  H.~Zhu, A.~Gupta, P.~Abbeel \emph{et~al.}, ``Soft actor-critic algorithms and
  applications,'' \emph{arXiv preprint arXiv:1812.05905}, 2018.

\end{thebibliography}


\end{document}